\documentclass[11pt]{llncs}

\usepackage{amssymb,amsmath,amscd,latexsym}
\usepackage{macro}
\usepackage{subfig,wrapfig}
\usepackage{tikz}
\usepackage{mytechrep}
\usetikzlibrary{arrows,shapes}

\begin{document}
\title{Best-Effort Strategies for Losing States\\
{\small Technical Report}}

\tikzstyle{p1}=[circle, draw=black, semithick]
\tikzstyle{p2}=[rectangle, draw=black, semithick]

\author{Marco Faella}
\institute{Computer Science Division\\
Physics Department\\
Universit\`a di Napoli ``Federico II'', Italy}

\maketitle

\newif
  \iflong
  \longfalse
\newif
  \ifshort
  \shorttrue

\begin{abstract}
We consider games played on finite graphs,
whose goal is to obtain a trace belonging to a given set
of winning traces.
We focus on those states from which Player~1 cannot
force a win.
We explore and compare several criteria
for establishing what is the preferable behavior 
of Player~1 from those states.

Along the way, we prove several results of theoretical
and practical interest, such as a characterization
of admissible strategies, which also provides
a simple algorithm for computing such strategies 
for various common goals, and the equivalence between
the existence of positional winning strategies and 
the existence of positional subgame perfect strategies.
\end{abstract}

\section{Introduction}

Games played on finite graphs have been widely investigated
in Computer Science, with applications including 
controller synthesis \cite{PnueliRosner89,ALW,MalerGames,emsoft05},
protocol verification~\cite{kremer01,trust07}, 
logic and automata theory~\cite{EmersonJutla91,Zielonka98}, and
compositional software verification~\cite{EMSOFT01}.

In such games, we are given a finite graph,
whose set of states is partitioned into Player~1 and Player~2 states, 
and a \emph{goal}, which is a set of infinite sequences of states.
The game consists in the two players taking turns at picking a successor state,
giving rise to an ever increasing and eventually infinite sequence of states.
A (deterministic) \emph{strategy} for a player is a function that, 
given the current history of the game (a finite sequence of states), 
chooses the next state.
A state $s$ is said to be \emph{winning} if there exists a strategy that
guarantees victory regardless of the moves of the adversary, if the game starts in $s$. 
A state that is not winning is called \emph{losing}.

The main algorithmic concern of the classical theory of these games 
is determining the set of winning states.
In this paper, we shift the focus to losing states,
since we claim that many applications would benefit from a theory of best-effort
strategies which allowed Player~1 to play in a rational way even from losing states.

For instance, many game models correspond to real-world problems
which are not really competitive: the game is just a tool which
enables to distinguish internal from external non-determinism.
In practice, the behavior of the adversary may turn out to be random, or even cooperative.
A strategy of Player~1 which does not ``give up'', but rather 
tries its best at winning, may in fact end up winning,
even starting from states that are theoretically losing.
For instance, such is the case in~\cite{emsoft05},
where games are used to model the interaction between the scheduler
of an operating system and the applications being scheduled.

In other cases, the game is an over-approximation of reality,
giving to Player~2 a wider set of capabilities (i.e., moves in the game)
than what most adversaries actually have in practice.
Again, a best-effort strategy for Player~1 can thus often lead to victory,
even against an adversary which is strictly competitive.

In this paper, we consider and compare several alternative definitions
of best-effort strategies.
As a guideline for our investigation, we take the application domain
of automated verification and synthesis of open systems.
Such domain is characterized by the fact that, once good strategies for
a game have been found, they are intended to be actually implemented
in hardware or software.
As a consequence, we tend to favor best-effort criteria that
are as discriminating (i.e., specific) as possible while still 
giving rise to efficient strategies.
While having such application domain in mind, we still put the main 
focus of the present work on theoretical issues, leaving to future
investigations a discussion of what criterion is more suitable
for any specific application.

\subsubsection{Best-effort strategies.}

The classical definition of what a ``good'' strategy is states
that a strategy is \emph{winning} if it guarantees victory
whenever the game is started in a winning state~\cite{Thomasgames95}.
This definition does not put any burden on a strategy
if the game starts from a losing state.
In other words, if the game starts from a losing state,
all strategies are considered equivalent.

A first refinement of the classical definition is a slight modification
of the game-theoretic notion of \emph{subgame-perfect equilibrium}~\cite{OsborneRubinstein}.
Cast in our framework, this notion states that a strategy is
good if it enforces victory whenever the game history is such that
victory can be enforced.
We call such strategies \emph{strongly winning}, 
to avoid confusion with the use of subgame (and subarena) which is common
in computer science~\cite{Zielonka98}.
It is easy to see that this definition captures the intuitive idea
that a good strategy should ``enforce victory whenever it can''
better than the classical one.

\begin{wrapfigure}[9]{r}{5cm}
\centering
    \begin{tikzpicture}[node distance=1.5cm, auto, bend angle=30, shorten >=2pt, shorten <=2pt]
      \node[p1] (s0)               {$s_0$};
      \node[p2] (s1) [right of=s0] {$s_1$};
      \node[p2] (s2) [right of=s1] {$s_2$};

      \path[-stealth'] (s0) edge (s1)
                            edge [bend right] (s2)
                       (s1) edge [bend right] (s0)
                            edge (s2)
                       (s2) edge [loop right] ();
  \end{tikzpicture}
   \caption{A game where victory cannot be enforced.}
  \label{fig-optimal3}
\end{wrapfigure}
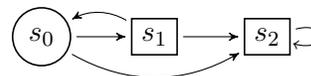

Next, consider games where victory cannot be enforced at any point
during the play. Take the B\"uchi game in Figure~\ref{fig-optimal3}
\footnote{Player~1 states are represented by circles and Player~2 states
by squares.},
whose goal is to visit infinitely often $s_0$.
No matter how many visits to $s_0$ Player~1 manages to make,
he will never reach a point where he can enforce victory.
Still, it is intuitively better for him to keep trying
(i.e., move to $s_1$) rather than give up (i.e., move to $s_2$).
To capture this intuition, we consider two further notions
of best-effort strategy. 

One such notion which is well known in the literature
is the one of optimal strategy in a particular type of game
called a \emph{Markov Decision Process} (MDP).
In such a game (also called 1.5-player game),
Player~2 plays according to a fixed distribution over successor states.
Thus, each strategy of Player~1 gives rise to a stochastic process,
i.e. to a distribution over infinite sequences of states.
Given a goal, one may then ask which is the strategy of Player~1 which
maximizes the probability of satisfying the goal.
There are known algorithms for solving this problem for various
classes of goals~\cite{bda95}.

The notion of optimal strategy suggests a first solution
to our original problem of dealing with losing states in a 2-player game.
We can assume that Player~2 plays uniformly at random and compute
an optimal strategy for Player~1.
Although this approach may be of interest to some cases,
it is worth considering alternative criteria,
which do not make as strong an assumption on the behavior of Player~2.
One such criterion is derived from the classical game-theoretic 
notion of \emph{dominance}~\cite{OsborneRubinstein}.
Given two strategies $\straa$ and $\straa'$ of Player~1,
we say that $\straa$ \emph{dominates} $\straa'$ 
if $\straa$ is always at least as good as $\straa'$, 
and better than $\straa'$ in at least one case.
Dominance induces a strict partial order over strategies,
whose maximal elements are called \emph{admissible} (or \emph{undominated}) strategies.
In Section~\ref{sec-compare}, we prove that a strategy is admissible if and only if
it is simultaneously strongly winning and cooperatively strongly winning
(i.e., strongly winning with the help of Player~2).

We claim that admissible strategies 
represent a convincing notion of best-effort strategy.
Similarly to optimality in a 1.5-player game, this notion
is goal-independent. Moreover, it does not make any assumption
on the behavior of Player~2.
Dominance can also be useful for multi-criteria optimization problems.
Differently from stochastic optimality, dominance gives rise 
to a \emph{partial} order over strategies, which leaves room for
another goal to be pursued.

\subsubsection{Memory.}

A useful measure for the complexity of a strategy consists in evaluating
how much memory it needs regarding the history of the game.
In the simplest case, a strategy requires no memory at all:
its decisions are based solely on the current state of the game.
Such strategies are called \emph{positional} or \emph{memoryless}~\cite{gimbert05}.
In other cases, a strategy may require the amount of memory that can be
provided by a finite automaton (\emph{finite memory}), or more~\cite{howmuch}.

The memory measure of a strategy is particularly important for the applications
that we target in this paper. Since we are interested in actually implementing
strategies in hardware or software, the simplest the strategy, the easiest
and most efficient it is to implement.

We devote Section~\ref{sec-memory} to studying the memory requirements for various
types of ``good'' strategies.
In particular, we prove that all goals that have positional winning strategies
also have positional \emph{strongly} winning strategies.
We also prove that for \emph{prefix-independent} positional goals, 
all positional winning strategies are automatically strongly winning.
The latter property is not valid for all positional goals. 
Thus, new algorithms are still needed to compute 
strongly winning strategies.

The situation is different for admissible strategies: there are games with
positional goals that have no positional admissible strategy.
Again, prefix-independent goals are particularly well-behaved:
those that admit positional winning strategies also admit positional
admissible strategies, as we show by presenting a simple algorithm
which computes positional admissible strategies for these goals.

\section{Definitions}
We treat games that are played by two players on a finite graph,
for an infinite number of turns.
The aim of the first player is to obtain an infinite trace
belonging to a fixed set of winning traces,
while the aim of the second player is the exact opposite.
In the literature, such games are termed
\emph{two-player}, \emph{turn-based}, 
\emph{qualitative} and \emph{zero-sum}.
The following definitions make this framework formal.

A \emph{game} is a tuple $G = (S_1, S_2, \delta, \col, F)$
such that: $S_1$ and $S_2$ are disjoint finite sets of states;
let $S = S_1 \cup S_2$, we have
that $\delta \subseteq S \times S$ is the transition relation
and $\col: S \to \nats$ is the coloring function,
where $\nats$ denotes the set of natural numbers including zero.
Finally, $F \subseteq  \nats^\omega$ is the \emph{goal}.
With an abuse of notation, we extend the coloring function 
to paths in the game graph, with the obvious meaning.
We assume that games are non-blocking, i.e. each state has at least
one successor in $\delta$.
Formally, $G' = (S_1, S_2, \delta', \col, F)$,
where $\delta' \cap (S_1 \times S) =         \delta \cap (S_1 \times S)$
and   $\delta' \cap (S_2 \times S) \subseteq \delta \cap (S_2 \times S)$.

A (finite or infinite) path in $G$ 
is a (finite or infinite) path in the directed graph $(S_1 \cup S_2, \delta)$.
If a finite path $\rho$ is a prefix of a finite or infinite path $\rho'$,
we also say that $\rho'$ \emph{extends} $\rho$.
We denote by $\first(\rho)$ the first state of a path $\rho$
and by $\last(\rho)$ the last state of a finite path $\rho$.

\subsubsection{Strategies.} 
A \emph{strategy} is a function $\stra: S^* \to S$ such that
for all $\trace \in S^*$, $(\last(\trace), \stra(\trace)) \in \delta$.
Our strategies are deterministic, 
or, in game-theoretic terms, \emph{pure}.
We denote $\Stra_G$ the set of all strategies in $G$.
We do not distinguish \emph{a priori} between strategies of Player~1 and Player~2.
However, for sake of clarity, we write $\straa$ for a strategy that should
intuitively be interpreted as belonging to Player~1,
and $\strab$ for the (rare) occasions when a strategy of Player~2 is needed.

Consider two strategies $\straa$ and $\strab$, and a finite path $\rho$,
and let $n = | \rho |$.
We denote by $\outc_G(\rho, \straa, \strab)$ the unique infinite path
$s_0 s_1 \ldots$ such that \emph{(i)} $s_0 s_1 \ldots s_{n-1} = \rho$,
and \emph{(ii)} for all $i \geq n$, $s_i = \straa(s_0 \ldots s_{i-1})$ if
$s_{i-1} \in S_1$ and $s_i = \strab(s_0 \ldots s_{i-1})$ otherwise.
We set 
$\outc_G(\rho, \straa) = \bigcup_{\strab\in\Stra_G} \outc_G(\rho, \straa, \strab)$.
For all $s \in S$ and $\rho \in \outc_G(s, \straa)$, 
we say that $\rho$ is \emph{consistent} with $\straa$. 
Similarly, we say that $\outc_G(s, \straa, \strab)$ is consistent
with $\straa$ and $\strab$.
We extend the definition of consistency from infinite paths 
to finite paths in the obvious way.

A strategy $\stra$ is \emph{positional} (or \emph{memoryless})
if $\stra(\trace)$
only depends on the last state of $\trace$.
Formally, for all $\trace,\trace' \in S^*$,
if $\last(\trace)=\last(\trace')$ then
$\stra(\trace) = \stra(\trace')$.

\subsubsection{Dominance.}
Given two strategies $\straa$ and $\strab$, and a finite path $\rho$,
we set $\val_G(\rho, \straa, \strab) = 1$ if $C(\outc_G(\rho, \straa, \strab)) \in F$,
and $\val_G(\rho, \straa, \strab) = 0$ otherwise.
Given two strategies $\straa$ and $\straa'$, 
we say that $\straa'$ \emph{dominates} $\straa$
if: \emph{(i)} for all $\strab \in \Strab_G$ and all $s\in S$, 
$\val_G(s, \straa', \strab) \geq \val_G(s, \straa, \strab)$, and
\emph{(ii)} there exist $\strab \in \Strab_G$ and $s\in S$ such that
$\val_G(s, \straa', \strab) > \val_G(s, \straa, \strab)$.

It is easy to check that dominance is an irreflexive, asymmetric
and transitive relation. 
Thus, it is a \emph{strict partial order} on strategies.

\subsubsection{Uniform stochastic games.}
In this presentation, 
a uniform stochastic game (USG) is syntactically equivalent to a game.
However, its semantics is different:
in a USG $G = (S_1,S_2,\delta,\col,F)$, 
in states in $S_1$ Player~1 chooses the successor state,
while in $S_2$ the successor is chosen according to a uniform distribution
over all successors.
A USG is thus a Markov Decision Process.
A state $s$ and a strategy $\straa$ induce a stochastic process
which is a Markov chain. 
We denote $\Prb_s^\straa(X)$ the probability of event $X$ in the Markov chain
generated by $\straa$ that starts at $s$.
For more information on stochastic processes, see~\cite{FilarVrieze97}.

Given a strategy $\straa$ and a state $s$, 
and assuming that $F$ is a measurable subset of $\nats^\omega$,
we denote $\val_G(s,\straa)$
the probability of winning using $\straa$ from state $s$,
that is $\Prb_s^\straa( \{ \rho \in S^\omega \mid \col(\rho) \in F \} )$.
We say that a strategy $\straa^*$ is \emph{optimal}
if it maximizes the probability of winning, from all states.
Formally, for all states $s$, we require
$\val_G(s,\straa^*) = \sup_{\straa\in\Straa_G} \val_G(s,\straa)$.

\subsection{Good Strategies} \label{sec-good}
In the following, unless stated otherwise,
we consider a fixed game $G = (S_1,S_2,\delta,\col,F)$ 
and we omit the $G$ subscript whenever the game is clear from the context.

Let $\rho$ be a finite path in $G$,
we say that a strategy $\straa$ is \emph{winning from $\rho$} if,
for all $\rho' \in \outc(\rho, \straa)$, the color sequence of $\rho'$
belongs to $F$.
We say that $\rho$ is \emph{winning} 
if there is a strategy $\straa$ which is winning from $\rho$.
The above definition extends to states, by considering them as length-1 paths.
A state that is not winning is called \emph{losing}.

Further, a strategy $\straa$ is \emph{cooperatively winning from $\rho$} if
there exists a strategy $\strab$ such that the color sequence of 
the unique infinite path
$\outc(\rho, \straa, \strab)$ belongs to $F$.
We say that $\rho$ is \emph{cooperatively winning}
if there is a strategy $\straa$ which is cooperatively winning from $\rho$.
Intuitively, a path is cooperatively winning if the two players together
can extend that path into an infinite path that satisfies the goal.
Again, the above definitions extend to states, 
by considering them as length-1 paths.

We can now present the following set of winning criteria.
Each of them is a possible definition of what a ``good'' strategy is.
\begin{itemize}
\item A strategy is \emph{winning} if it is winning from all winning states.
This criterion intuitively demands that strategies enforce victory whenever
the initial state allows it.
\item A strategy is \emph{subgame perfect} if it is winning from all winning paths.
This criterion states that a strategy should enforce victory whenever
the current history of the game allows it.
\item A strategy is \emph{strongly winning} if it is winning 
from all winning paths that are consistent with it.
\item A strategy is \emph{cooperatively winning} 
(in short, \emph{c-winning})
if it is cooperatively winning from all cooperatively winning states.
This criterion essentially asks a strategy to be winning with the help of Player~2. 
\item A strategy is \emph{cooperatively subgame perfect} 
(in short, \emph{c-perfect})
if it is cooperatively winning from all cooperatively winning paths.
\item A strategy is \emph{cooperatively strongly winning} 
(in short, \emph{cs-winning})
if it is cooperatively winning from all cooperatively winning paths
that are consistent with it.
\item A strategy is \emph{admissible} if there is no strategy that dominates it.
This criterion favors strategies that are maximal w.r.t. the partial order 
defined by dominance.
\item A strategy is \emph{optimal} if it is so 
in the corresponding uniform stochastic game.
This criterion endorses strategies that are optimal, assuming that the adversary
chooses her moves according to a uniform distribution.
\end{itemize}

The notions of winning, optimal and cooperatively winning strategies
are customary to computer scientists~\cite{Thomasgames95,ATL-FOCS97}.
The notion of subgame perfect strategy comes from
classical game theory~\cite{OsborneRubinstein}.
The introduction of the notion of strongly winning strategy
is motivated by the fact that in the target applications 
game histories that are inconsistent with the strategy of Player~1
cannot occur. Being strongly winning is strictly weaker than being 
subgame perfect.
In addition, there are games for which there is a positional strongly winning
strategy, but no positional subgame perfect strategy
(see Figure~\ref{fig-subgame-needs-mem} in the Appendix).
The term ``strongly winning'' seems appropriate
since this notion is a natural strengthening 
of the notion of winning strategy.

We say that a goal $F$ is \emph{positional}
if for all games $G$ with goal $F$, there is a positional strategy
that is winning in $G$.
Recently, necessary and sufficient conditions for a goal to be positional
were identified~\cite{gimbert05}.

\section{Comparing Winning Criteria} \label{sec-compare}

In this section, we compare the winning criteria 
presented in Section~\ref{sec-good}, taking as the main reference
the definition of winning strategy.
Figure~\ref{fig-compare} summarizes the relationships between
the winning criteria under consideration.
We start by stating the following basic properties.
\begin{lemma} \label{lem-props}
The following properties hold:
\begin{enumerate}
\item \label{item-a} all strongly winning strategies are winning, but not vice versa;
\item \label{item-b} all subgame perfect strategies are strongly winning, but not vice versa;
\item \label{item-d} all cs-winning strategies are c-winning, but not vice versa;
\item \label{item-e} all c-perfect strategies are cs-winning, but not vice versa;
\item \label{item-c} all games have a winning 
(respectively, strongly winning, subgame perfect, c-winning,
cs-winning, c-perfect, optimal, admissible) strategy.
\end{enumerate}
\end{lemma}
\begin{proof}
The containments stated in (\ref{item-a}) and (\ref{item-b}) are obvious by definition.
The fact that those containments are strict is easily proved 
by the game in Figure~\ref{fig-win-not-strongly} in the Appendix.
Similarly, statements (\ref{item-d}) and (\ref{item-e}) follow from the definitions
and from the example in Figure~\ref{fig-cwin-not-strongly} in the Appendix.

Regarding statement (\ref{item-c}), 
the existence of a winning (respectively, strongly winning,
subgame perfect, c-winning, cs-winning, c-perfect, optimal) strategy
is obvious by definition.
The existence of an admissible strategy can be derived
from Theorem~11 from~\cite{admissibility}.
\qed
\end{proof}

\subsection{Strongly Winning Strategies}

It is easy to check that winning strategies need not be strongly winning.
Here, we give a sufficient condition that a goal can satisfy to ensure
that all winning strategies are strongly winning.

A goal $F$ is \emph{shrinkable} iff
for all $c \rho \in F$, with $c \in \nats$ and $\rho \in \nats^\omega$,
$\rho \in F$.
A goal $F$ is \emph{extensible} iff
for all $\rho \in F$ and all $c \in \nats$, $c \rho \in F$.
A goal $F$ is \emph{prefix-independent} iff
it is both shrinkable and extensible.
Examples of common prefix-independent goals include
B\"uchi, co-B\"uchi and parity goals.

\begin{lemma} \label{lem-worsen}
If a goal $F$ is shrinkable, then, for all games with goal $F$,
all winning paths end in a winning state, and
all c-winning paths end in a c-winning state.
\end{lemma}
\begin{proof}
We prove the statement for winning paths, as the one regarding
c-winning paths can be proved along similar lines.
Let $G$ be a game with a shrinkable goal $F$, 
let $\rho = s_0 \ldots s_n$ be a winning path and
let $\straa$ be a strategy which is winning from $\rho$.
Consider all infinite paths that extend $\rho$ and are consistent with
$\straa$. These paths all satisfy the goal $F$. 
If we remove the prefix $\rho$ from these paths, 
they still all satisfy $F$, since $F$ is shrinkable.
Consider the strategy $\straa'$ defined by:
for all $\pi \in S^*$,
$$
\straa'(\pi) = 
\begin{cases}
\straa(s_0 \ldots s_{n-1} \pi) &\text{if $\first(\pi) = s_n$,} \\
\text{arbitrarily defined}     &\text{otherwise.}
\end{cases}
$$
It is immediate that $\straa'$ is winning from $s_n$ and therefore
$s_n$ is a winning state.
\qed
\end{proof}
The following corollary states that if a goal is shrinkable,
winning strategies confine the game in the winning region.

\begin{corollary} \label{lem-trapping}
If a goal is shrinkable, for all winning strategies $\straa$,
and for all finite paths $\rho$ consistent with $\straa$,
if $\first(\rho)$ is winning then $\last(\rho)$ is winning.
\end{corollary}

\begin{theorem} \label{thm-pindependent}
If a goal $F$ is prefix-independent, then, for all games with goal $F$,
all positional winning strategies are strongly winning,
and all positional c-winning strategies are cs-winning.
\end{theorem}
\begin{proof}
We prove the statement for positional winning strategies, 
as the one regarding positional c-winning strategies 
can be proved along similar lines.
Let $G$ be a game with goal $F$, and let $\straa$ be a winning strategy for $G$.
Let $\rho = s_0 \ldots s_n$ be a winning path which is consistent with $\straa$.
By Lemma~\ref{lem-worsen}, $s_n$ is a winning state.
Let $\rho' = s_0 \ldots s_n s_{n+1} \ldots$ 
be an infinite path which extends $\rho$ and is consistent with $\straa$.
Since $\straa$ is positional and winning from $s_n$, 
$\col(s_{n} s_{n+1} \ldots) \in F$.
Since $F$ is extensible, $\col(\rho') \in F$. 
Therefore, $\straa$ is winning from $\rho$.
Being $\rho$ generic, we conclude that $\straa$ is strongly winning.
\qed
\end{proof}
The positionality assumption is necessary in the above result.
For a prefix-independent goal, it is easy to devise
winning strategies that are not positional and not strongly winning.
On the other hand, being prefix-independent is not necessary for ensuring that 
all positional winning strategies are strongly winning. 
For instance, safety and reachability goals are not prefix-independent,
but they ensure said property.
Finally, simple examples show that 
neither shrinkability nor extensibility alone can replace prefix-independence
in the assumptions of the above result.

\begin{figure}[htp]
   \centering
   \begin{tikzpicture}[scale=0.70]
     \filldraw[fill=gray!15,very thick]  
        (1.5,0) arc (0:60:3cm and 1.5cm) 
                arc (120:240:3cm and 1.5cm)
                arc (300:360:3cm and 1.5cm);     
     \draw   (-2,0) ellipse (4cm and 2cm);
     \draw (-1.5,0) ellipse (3cm and 1.5cm);
     \draw   (-1,0) ellipse (2cm and 1cm);
     \draw    (2,0) ellipse (4cm and 2cm);
     \draw  (1.5,0) ellipse (3cm and 1.5cm);
     \draw    (1,0) ellipse (2cm and 1cm);
     \draw    (0,2) ellipse (1cm and 2.5cm);

     \draw    (6,0) arc (0:30:4cm and 2cm)   -- ++(30:0.3cm) node[right] {C-Winning};
     \draw  (4.5,0) arc (0:40:3cm and 1.5cm) -- ++(40:1.8cm) node[right] {CS-Winning};
     \draw    (3,0) arc (0:60:2cm and 1cm)   -- ++(60:2.5cm) node[right] {C-Perfect};

     \draw   (-6,0) arc (180:150:4cm and 2cm)   -- ++(150:0.3cm) node[left] {Winning};
     \draw (-4.5,0) arc (180:140:3cm and 1.5cm) -- ++(140:1.8cm) node[left] {Strongly Winning};
     \draw   (-3,0) arc (180:120:2cm and 1cm)   -- ++(120:2.5cm) node[left] {Subgame Perfect};

     \draw   (0,0) ++(90:4.5cm) -- ++(90:0.3cm) node[above] {Optimal};
     \path   (0,0) node {\bf Admissible};
     \end{tikzpicture}
   \caption{Comparing winning criteria.}
  \label{fig-compare}
\end{figure}
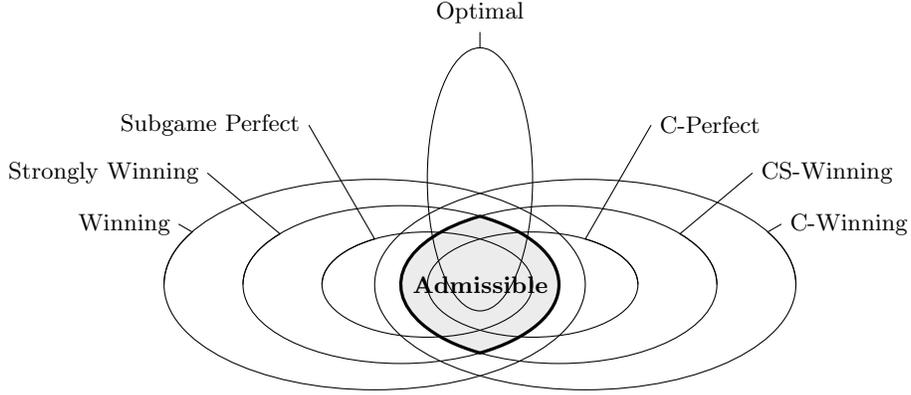

\subsection{Admissible Strategies}

In this section we provide a characterization of admissibility
in terms of the simpler criteria of strongly winning 
and cooperatively strongly winning.
Such characterization will be useful to derive further properties
of admissible strategies.
We start with a trivial property of winning paths.

\begin{lemma} \label{lem-prefix}
If $s_0 s_1 \ldots s_n$ is a winning path and $s_{n-1} \in S_1$,
then $s_0 s_1 \ldots s_{n-1}$ is a winning path.
\end{lemma}
We can now prove the main theorem of this section.

\begin{theorem} \label{thm-undominated}
A strategy is admissible
if and only if it is strongly winning and 
cooperatively strongly winning.
\end{theorem}
\begin{proof}
For the ``if'' part, let $\straa^*$ be a strategy
which is both strongly winning and cs-winning. 
Assume that there is 
a strategy $\straa$ that is better than 
$\straa^*$ in at least one case.
In particular, let $s$ be a state and $\strab$ be a strategy of Player~2 such that
$\val(s, \straa^*, \strab) = 0$ and $\val(s, \straa, \strab) = 1$.
We can build a Player~2 strategy $\strab'$ such that
$\val(s, \straa^*, \strab') = 1$ and $\val(s, \straa, \strab') = 0$,
thus proving that $\straa$ does not dominate $\straa^*$.

Let $\alpha u$ be the longest prefix common to both $\outc(s, \straa, \strab)$
and $\outc(s, \straa^*, \strab)$.
Precisely,
let $\outc(s, \straa,   \strab) = \alpha u v   \ldots$
and $\outc(s, \straa^*, \strab) = \alpha u v^* \ldots$,
where $\alpha \in S^*$, $u,v,v^* \in S$, and $v \neq v^*$.
Clearly, it must be $u \in S_1$.
We have that $\alpha u$ is not a winning path because $\straa^*$ is strongly winning
and still $\val(s, \straa^*, \strab) = 0$.
By Lemma~\ref{lem-prefix}, $\alpha u v$ is not a winning path either.
We also have that $\alpha u$ is c-winning, because $\val(s, \straa, \strab) = 1$.
Since $\straa^*$ is cs-winning, $\alpha u v^*$ is also c-winning.

Then, we define the Player~2 strategy $\strab'$ as follows. 
We let $\strab'$ coincide with $\strab$ on all finite paths
that are prefixes of $\alpha u$.
On all paths that are extensions of $\alpha u v$, 
we let $\strab'$ behave in such a way to ensure that $\val(s, \straa, \strab') = 0$.
This is possible because $\alpha u v$ is not a winning path.
On all paths that are extensions of $\alpha u v^*$, 
we let $\strab'$ behave in such a way to ensure that $\val(s, \straa^*, \strab') = 1$.
This is possible because $\alpha u v^*$ is a c-winning path
and $\straa^*$ is cs-winning.
We conclude that $\straa$ does not dominate $\straa^*$.
Therefore, no strategy dominates $\straa^*$.

Next, we prove the ``only if'' part.
Let $\straa^*$ be an admissible strategy.
By contradiction, assume that $\straa^*$ is not strongly winning.
Therefore, there exist a winning path $\rho$ and a strategy $\strab^*$
such that $\rho$ is consistent with $\straa^*$ and $\strab^*$
and $\val(\rho,\straa^*,\strab^*) = 0$.
Now, let $\straa$ be a strategy that is winning from $\rho$.
Define another Player~1 strategy $\straa'$ as follows,
for all $\pi \in S^*$:
$$
\straa'(\pi) =
\begin{cases}
\straa(\pi)   &\text{if $\pi$ extends $\rho$,} \\
\straa^*(\pi) &\text{otherwise.}
\end{cases}
$$
We show that $\straa'$ dominates $\straa^*$, which is a contradiction.
Take any Player~2 strategy $\strab$.
If $\straa'$ and $\strab$ together do not give rise to the finite path $\rho$,
$\straa'$ behaves exactly like $\straa^*$.
If $\straa'$ and $\strab$ together do give rise to the path $\rho$,
from that point on $\straa'$ behaves like $\straa$, and therefore ensures victory.
This proves that $\straa'$ always performs at least as well as $\straa^*$.
Finally, there is a case where $\straa'$ performs better than $\straa^*$: 
we have $\val(\first(\rho),\straa', \strab^*) = 1$ and 
$\val(\first(\rho),\straa^*, \strab^*) = 0$,
which concludes the proof of the contradiction.

Next, we show that $\straa^*$ is cs-winning.
Let $\rho$ be a c-winning path and assume by contradiction that,
for all strategies $\strab$, $\val(\rho, \straa^*, \strab) = 0$.
Let $\straa^\bullet, \strab^\bullet$ be a pair of strategies such that
$\val(\rho, \straa^\bullet, \strab^\bullet) = 1$.
Let $\straa$ be a strategy that behaves like $\straa^*$,
except that for all paths extending $\rho$ it behaves like $\straa^\bullet$.
We prove that $\straa$ dominates $\straa^*$.
If the path $\rho$ is not formed during the game, 
$\straa$ behaves exactly like $\straa^*$.
If the path $\rho$ is formed, $\straa^*$ loses with certainty,
while $\straa$ wins in at least one case, namely against $\strab^\bullet$.
This proves that $\straa$ dominates $\straa^*$, which is a contradiction.
\qed
\end{proof}

\subsection{Optimal Strategies} \label{sec-compare-optimal}

A natural notion of good strategy in a losing
state is given by optimal strategies in the corresponding
uniform stochastic game (USG).
In this section, we explore the relationship between optimality
and the other winning criteria.

First of all, a trivial example (see Figure~\ref{fig-optimal1} in the Appendix)
shows that neither winning strategies
nor strongly winning strategies nor even admissible strategies need be optimal.
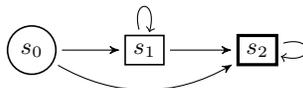
\begin{figure}[htp]
\centering
    \begin{tikzpicture}[node distance=1.5cm, auto, bend angle=30, shorten >=2pt, shorten <=2pt]
      \node[p1] (s0)               {$s_0$};
      \node[p2] (s1) [right of=s0] {$s_1$};
      \node[p2, very thick] (s2) [right of=s1] {$s_2$};

      \path[-stealth'] (s0) edge (s1)
                            edge [bend right] (s2)
                       (s1) edge [loop above] ()
                            edge (s2)
                       (s2) edge [loop right] ();
  \end{tikzpicture}
  \caption{A game showing that optimal strategies need not be winning.}
  \label{fig-optimal2}
\end{figure}
The game in Figure~\ref{fig-optimal2} shows that the converse
also holds, i.e. optimal strategies need not be winning (and therefore
neither strongly winning nor admissible).
Suppose that the goal is to reach state $s_2$.
In the corresponding USG, both strategies of Player~1 are optimal,
yielding a probability 1 of winning.
However, only the strategy that goes directly to $s_2$ is winning.
So, in this case winning proves to be finer (more discriminating)
than optimal.

Finally, one might wonder whether 
starting from losing states optimal strategies are always admissible.
The game in Figure~\ref{fig-optimal3},
already presented in the introduction, settles this question in the negative. 
Suppose the goal is to visit infinitely often $s_0$.
All states and all paths in the game are losing.
In the corresponding USG, all strategies have value 0.
However, the strategy that always picks $s_1$ dominates
all others, since it has a chance of winning (namely,
against the strategy that always picks $s_0$).

\section{Memory} \label{sec-memory}

In this section, we study the amount of memory required by
``good'' strategies for achieving different kinds of goals.
We are particularly interested in identifying those goals 
which admit positional good strategies, because positional
strategies are the easiest to implement.
The problem of identifying goals
that admit positional winning strategies and positional optimal strategies
is well studied in the literature.
Here, we focus on the winning criteria of
``strongly winning'' and ``admissible''.

\subsection{Positional Strongly Winning Strategies}

For a game $G = (S_1, S_2, \delta, \col, F)$ 
and a path $\rho = s_0 \ldots s_n$ in $G$,
define $\detach(G,\rho)$ as the game obtained from $G$
by adding a copy of the path $\rho$ to it as a chain of new states
ending in the original state $s_n$.
Formally, $\detach(G,\rho) = (S_1,S_2',\delta',\col', F)$,
where $S_2' = S_2 \cup \{ s_0', s_1', \ldots, s_{n-1}' \}$
and $s_0', s_1', \ldots, s_{n-1}'$ are new distinct states not belonging to $S_2$ nor to $S_1$.
Then, $(s,t) \in \delta'$ iff either 
\emph{(i)} $(s,t) \in \delta$, or
\emph{(ii)} $s=s_i'$ and $t=s_{i+1}'$, or 
\emph{(iii)} $s=s_{n-1}'$ and $t=s_n$.
Finally, the color labeling is defined by: 
$$\col'(s) = \begin{cases}
  \col(s_i) &\text{if $s=s_i'$ for some $i \in \{0,\ldots,n-1\}$,} \\
  \col(s)   &\text{otherwise.}
\end{cases}
$$
The detach operation adds some states to the game, but the new
states are not reachable from any of the old states.
As a consequence, if games $G$ and $\detach(G,\rho)$ start from an old state $s$,
they are indistinguishable to both players.
The following lemma formalizes this observation, 
by stating in particular that the detach operation preserves the winning property of paths.

Notice that in the following we commit a slight abuse in language 
by identifying positional strategies in $G$ and in $\detach(G,\rho)$.
This is justified by the fact that the new states introduced in $\detach(G,\rho)$
only have one successor, thus giving no more choices to either player.

\begin{lemma} \label{lem-detach1}
For all strategies $\straa$, and finite paths $\rho, \rho'$ in $G$,
$\straa$ is winning from $\rho$ in $G$ if and only if
 it is winning from $\rho$ in $\detach(G,\rho')$.
\end{lemma}
The following lemma provides the very reason for introducing the detach operation
in the first place. Consider a winning path in the original game $G$.
In order for this path to occur in a play, we may need the collaboration of Player~2. 
On the other hand,
detaching this path allows us to force its occurrence, with no help from Player~2.
As a consequence, the initial state $s_0'$ of the detached path is a winning state
in the detached game, as stated by the following lemma.

\begin{lemma} \label{lem-detach2}
Let $\rho = s_0 \ldots s_n$ be a path in $G$
and let $G' = \detach(G,\rho)$.
Let $s_0'$ be the new state added to $G'$ in correspondence to $s_0$.
If $\rho$ is a winning path in $G$, 
then $s_0'$ is a winning state in $G'$.
\end{lemma}
\begin{proof}
Let $\straa$ be a strategy that is winning from $\rho$ in $G$,
and let $\rho' = s_0' \ldots s_{n-1}' s_n$ be the detachment of $\rho$
added to $G'$.
Define a strategy $\straa'$ in $G'$ as follows.
For all finite paths $\pi$ in $G'$, if $\pi$ extends $\rho'$,
i.e. $\pi = \rho' \pi'$, set $\straa'(\pi) = \straa( \rho \pi')$.
For all other paths in $G'$, $\straa'$ chooses an arbitrary move
available to Player~1.
It is easy to check that the set of infinite paths in $G$ 
which are consistent with $\straa$ and extend $\rho$
is color-equivalent to the set of infinite paths in $G'$
which are consistent with $\straa'$ and start in $s_0'$.
Therefore, $\straa'$ is winning from $s_0'$ in $G'$.
\qed
\end{proof}
With the help of the previous lemmas, we are ready to prove the following.

\begin{theorem}
A goal is positional if and only if it admits
positional subgame perfect strategies.
\end{theorem}
\begin{proof}
The ``if'' part being obvious by definition, assume that the goal is positional.
Let $G$ be a game and let $W$ be the set of winning paths of $G$.
$W$ may be infinite but it is certainly countable.
Consider any ordering of $W$ into $\rho_0, \rho_1, \ldots$.
Consider the sequence of games $(G_i)_{i \geq 0}$ defined
by $G_0 = G$ and $G_{i+1} = \detach(G_i, \rho_i)$.
Additionally, consider the sequence of strategies $(\straa_i)_{i \geq 0}$
defined by: $\straa_0$ is any positional winning strategy in $G_0$, and
$$ \straa_{i+1} = \begin{cases}
   \straa_i &\text{if $\straa_i$ is winning in $G_{i+1}$,} \\
   \text{any positional winning strategy in $G_{i+1}$} &\text{otherwise}.
\end{cases}
$$
We prove that the sequence $(\straa_i)_{i \geq 0}$ converges to a strategy
$\straa^*$ within a finite number of steps.
Specifically, we show that once a strategy $\straa$ occurs in the sequence
and then is replaced by another one, it will not occur 
at any later point in the sequence. This fact, together with the fact that
the number of positional strategies is finite, leads to the convergence
of the sequence.
Assume by contradiction that there exist indices $0 \leq a < b$ such that
$\straa_a \neq \straa_{a+1} \neq \straa_{b}$ and $\straa_a = \straa_b$.
Since $\straa_a \neq \straa_{a+1}$, $\straa_a$ is not winning in $G_{a+1}$.
Therefore, by repeated application of Lemma~\ref{lem-detach1}, $\straa_a$
(or equivalently, $\straa_b$)
cannot be winning in $G_b$, which is a contradiction.

Next, we prove that $\straa^*$ is subgame perfect in $G$.
Let $\rho = s_0 \ldots s_n$ be a winning path in $G$.
There is an index $j>0$ such that $G_j = \detach(G_{j-1}, \rho)$.
Let $s_0'$ be the ``copy'' of $s_0$ added to $G_j$ by the detach operation.
Since $\rho$ is a winning path in $G$, by Lemma~\ref{lem-detach1}
it is still winning in $G_{j-1}$, and by Lemma~\ref{lem-detach2}
$s_0'$ is a winning state in $G_j$.
Therefore, $\straa_j$ is winning from $s_0'$.
Since $\straa^*$ is the ultimate value for the sequence of strategies,
$\straa^*$ is also winning from $s_0'$.
Take an infinite path $\rho'$ in $G$,
which extends $\rho$ and that after $\rho$ is consistent with $\straa^*$.
If we replace in $\rho'$ the prefix $\rho$ with $s_0' \ldots s_{n-1}' s_n$, 
we obtain an infinite path $\rho''$ in $G_j$,
which is consistent with $\straa^*$ (because $\straa^*$ is positional)
and color-equivalent to $\rho'$.
Since $\straa^*$ is winning from $s_0'$, $\col(\rho'') \in F$ 
and therefore $\col(\rho') \in F$, which concludes the proof.
\qed
\end{proof}
Notice, however, that there are games with non-positional goals
that have a positional winning strategy but no positional subgame perfect strategy.
An immediate consequence of the above theorem is the following.

\begin{corollary}
A goal is positional if and only if it admits strongly winning strategies.
\end{corollary}

\subsection{Positional Admissible Strategies}

Since all admissible strategies are winning, admissible
strategies require at least as much memory as winning strategies.
The following example shows that there are goals which admit
positional winning strategies, but do not admit positional
admissible strategies.

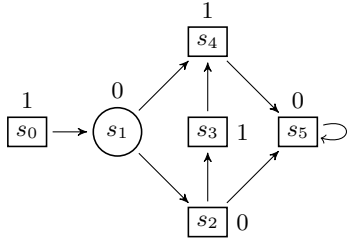
\begin{figure}[htp]
\begin{center}
    \begin{tikzpicture}[remember picture, node distance=1.2cm, auto, bend angle=20, shorten >=2pt, shorten <=2pt]
      \node[p2] (s0) [label={above:$1$}]              {$s_0$};
      \node[p1] (s1) [label=above:$0$,        right of=s0] {$s_1$};
      \node[p2] (s3) [label={right:$1$}, right of=s1] {$s_3$};
      \node[p2] (s2) [label=right:$0$,        below of=s3] {$s_2$};
      \node[p2] (s4) [label={above:$1$}, above of=s3] {$s_4$};
      \node[p2] (s5) [label=above:$0$,        right of=s3] {$s_5$};

      \path[-stealth'] (s0) edge (s1)
                       (s1) edge (s2)
                            edge (s4)
                       (s2) edge (s3)
                            edge (s5)
                       (s3) edge (s4)
		       (s4) edge (s5)
                       (s5) edge [loop right] ();
  \end{tikzpicture} 
  \end{center}
   \caption{Positional goals need not admit positional admissible strategies.}
  \label{fig-mem1}
\end{figure}

\begin{example}
Consider the goal which requires to visit at least twice color $1$
(in $\ltl$ notation, $\phi = \Diamond ( 1 \wedge \Next \Diamond 1)$).
Such goal is positional, as it satisfies the criterion of~\cite{gimbert05}.

Consider the game in Figure~\ref{fig-mem1}, with the above goal $\phi$.
In the figure, the color of each state appears next to it or above it.
One can easily check that states $s_0$ and $s_3$ are winning, while
all others are losing.
The only choice for Player~1 occurs in $s_1$, where he can choose
between $s_2$ and $s_4$. If the game starts in $s_0$, in order to ensure
victory, Player~1 must choose $s_4$ after $s_1$.
On the other hand, suppose that the game starts in $s_1$.
State $s_1$ is cooperatively winning, since both players can cooperate
and achieve victory by following the path $s_1 s_2 s_3 s_4 s_5^\omega$.
By Theorem~\ref{thm-undominated} for a strategy to be admissible
it has to be both winning and cooperatively winning. 
Thus, Player~1 must choose $s_2$ after $s_1$, if the game started in $s_1$ itself.
In conclusion, the choice of any admissible strategy from $s_1$
depends on the past history of the game.
\end{example}

In the following, we give a sufficient condition to ensure
that a goal admits positional admissible strategies.
We start with a simple algorithm, which we later show
yields positional admissible strategies for many goals of interest.

\subsubsection{Computing positional admissible strategies.}

Suppose that we are given a game $G$ with a positional goal $F$,
and that we have an algorithm for computing the set of winning states
and a positional winning strategy for all games with goal $F$.
Consider the following procedure.

\begin{enumerate}
\item Compute the set of winning states $\win$ and a positional winning strategy $\straa$ for $G$.
\item Remove from $G$ the edges of Player~1 
      which start in $\win$ and do not belong to $\straa$.
\item In the resulting game, 
      compute and return a positional cooperatively winning strategy.
\end{enumerate}

In general, this procedure may return a strategy
that is neither winning nor cooperatively winning.
However, in the following we show that indeed it returns
a strategy that is both winning and cooperatively winning
in many cases of interest.

As far as the complexity of the procedure is concerned, 
assuming the usual graph-like adjacency-list representation for games,
we obtain the same asymptotical complexity
as finding a positional winning strategy for $F$.
In particular, step 3 can easily be performed by attributing all states to Player~1
and then running the algorithm for a positional winning strategy.

In the following, we prove that for prefix-independent positional goals
the above procedure returns an admissible strategy.

\begin{theorem}
If a goal is positional and prefix-independent, 
then it admits positional admissible strategies.
\end{theorem}
\begin{proof}
Let $G$ be a game with goal $F$, where $F$ is positional and prefix-independent.
Let $\win$ be the set of winning state of $G$,
and consider the application of the above procedure to $G$.
Let $\straa_1$ be the strategy computed at step~1
of the procedure, and $\straa_3$ the output strategy.
Clearly, $\straa_3$ is positional.
We prove that it is also admissible.
According to Theorems~\ref{thm-pindependent} and~\ref{thm-undominated}, 
it is sufficient to prove that $\straa_3$ is winning and cooperatively winning.

By Corollary~\ref{lem-trapping}, whenever the game starts in a winning state,
$\straa_1$ confines the game in $\win$.
Since $\straa_3$ coincides with $\straa_1$ on $\win$,
$\straa_3$ is winning.

Next, let $G_2$ be the game built at Step~2 of the procedure.
We know that $\straa_3$ is cooperatively winning in $G_2$.
We now prove that $\straa_3$ is cooperatively winning in $G$.
Let $s$ be a cooperatively winning state in $G$.
We prove that $s$ is cooperatively winning in $G_2$ as well.
Let $\straa,\strab$ be a pair of strategies in $G$ 
such that $C(\outc(s,\straa,\strab)) \in F$.
We consider the following two cases:
\emph{(1)} the infinite path $\outc(s,\straa,\strab)$ does not
visit $\win$; then, strategies $\straa$ and $\strab$ are valid in $G_2$
and thus $s$ is cooperatively winning in $G_2$;
\emph{(2)} let $\rho$ be the shortest prefix of $\outc(s,\straa,\strab)$ which ends in $\win$;
then, consider a new strategy $\straa'$ which behaves as follows, for all $\pi\in S^*$:
$$
\straa'(\pi)=
\begin{cases}
\straa_1(\pi) &\text{if $\pi$ extends $\rho$,} \\
\straa(\pi)   &\text{otherwise.}
\end{cases}
$$
Consider strategies $\straa'$ and $\strab$ playing together in $G_2$.
At first, $\straa'$ behaves like $\straa$ and so the path $\rho$ is built.
Then, $\straa'$ behaves like the winning strategy $\straa_1$.
Eventually, the infinite path $\rho \rho'$ is built,
where $C( \last(\rho) \rho') \in F$.
Since $F$ is extensible, $C(\rho \rho')\in F$ and 
thus $s$ is cooperatively winning in $G_2$.
\qed
\end{proof}
This result proves that common goals such as B\"uchi, co-B\"uchi, and parity
all admit positional admissible strategies.
However, prefix-independence is not necessary 
for admitting positional admissible strategy.
Further, it is not necessary for the procedure to work either.
In particular, it is easy to prove that 
the procedure also returns an admissible strategy for reachability and safety goals.

\section{Conclusions and Future Work}

We studied and compared several criteria for establishing what are 
the preferable strategies for Player~1 in a given game.
Being winning, or rather strongly winning, intuitively 
seems to be a necessary prerequisite for any best-effort strategy.
Accordingly, optimality against a random opponent should be discarded
as it does not imply winning. 
However, for some applications it could be worth considering
a strategy which is primarily strongly winning, 
and optimal from paths that are not winning.
Admissibility, on the other hand, passes all our tests,
including the one of being efficiently computable for goals of interest.
On the other hand, in general admissibility incurs some memory cost.

This preliminary study leaves several open problems, including:
characterizing goals having positional admissible strategies,
finding algorithms for computing strongly winning and admissible strategies
for goals that are not prefix-independent. 

Aside from the goal-independent best-effort criteria studied in this paper,
it is worth noticing that goal-specific criteria offer an entirely
different range of possibilities, many of which remain unexplored 
in the literature.

\subsubsection{Acknowledgments.}
The author would like to thank Hugo Gimbert for a fruitful conversation
on the subject.

\newpage
\section*{Appendix: Additional Examples}

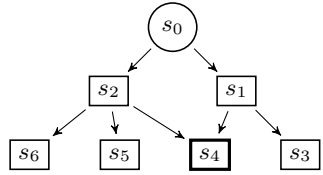
\begin{figure}[h]
\centering
    \begin{tikzpicture}[node distance=1.2cm, auto, bend angle=20, shorten >=2pt, shorten <=2pt]
      \node[p1] (s0) {$s_0$};
      \node[p2] (s1) [below right of=s0] {$s_1$};
      \node[p2] (s2) [below left of=s0] {$s_2$};
      \node[p2] (s3) [below right of=s1] {$s_3$};
      \node[p2,very thick] (s4) [left of=s3] {$s_4$};
      \node[p2] (s5) [left of=s4] {$s_5$};
      \node[p2] (s6) [left of=s5] {$s_6$};

      \path[-stealth'] (s0) edge (s1)
                            edge (s2)
                       (s1) edge (s3)
                            edge (s4)
                       (s2) edge (s4)
                            edge (s5)
                            edge (s6);
  \end{tikzpicture}
  \caption{Admissible strategies need not be optimal. The goal is to reach $s_4$. 
  All strategies of Player~1 are winning, strongly winning and admissible.
  On the other hand, only the strategy that chooses $s_1$ is optimal,
  leading to a probability of winning of $\frac{1}{2}$.}
  \label{fig-optimal1}
\end{figure}
\begin{figure}[h]
\centering
\begin{tikzpicture}[node distance=1.7cm, auto, shorten >=2pt, shorten <=2pt, bend angle=20]
      \node[p2] (s0) [label=below:$1$]              {$s_0$};
      \node[p1] (s1) [label=below:$2$, right of=s0] {$s_1$};

      \path[-stealth'] (s0) edge [bend right] (s1)
                            edge [loop above] ()
                       (s1) edge [thick,bend right] (s0)
                            edge [loop above] ();
\end{tikzpicture} 
\caption{
A game having a winning strategy which is not strongly winning,
and a strongly winning strategy which is not subgame perfect.
The goal is $\phi = 1 \wedge \Diamond \Box 2$.
The positional strategy consisting of going from $s_1$ to $s_0$ 
(the thick edge) is winning but not strongly winning.
The strategy that chooses $s_1$ when the current history contains
exactly one occurrence of $s_0$, and chooses $s_0$ otherwise
is strongly winning but not subgame perfect, due for instance
to the winning path $s_0 s_1 s_0 s_1$.}
\label{fig-win-not-strongly}
\end{figure}
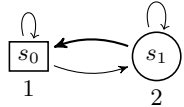
\begin{figure}[h]
\centering
\begin{tikzpicture}[node distance=1.7cm, auto, shorten >=2pt, shorten <=2pt, bend angle=20]
      \node[p1] (s0) [label=above:$1$]              {$s_0$};
      \node[p2] (s1) [label=above:$0$, right of=s0] {$s_1$};
      \node[p1] (s2) [label=above:$0$, right of=s1] {$s_2$};
      \node[p2] (s3) [label=below:$2$, node distance=1cm, below of=s1] {$s_3$};

      \path[-stealth'] (s0) edge              (s1)
                            edge [bend angle=45, bend left]  (s2)
                       (s1) edge [bend right] (s2)
                            edge              (s3)
                       (s2) edge [bend right] (s1)
                            edge [loop right] ();
\end{tikzpicture} 
\caption{
A game having a c-winning strategy which is not cs-winning,
and a cs-winning strategy which is not c-perfect.
The goal is $\phi = 1 \wedge \Diamond \Box 2$.
The only c-winning state is $s_0$.
The positional strategy $s_0 \rightarrow s_1$, $s_2 \rightarrow s_2$
is c-winning but not cs-winning. 
Consider the strategy that chooses $s_0 \rightarrow s_1$ and then
\emph{(i)} $s_2 \rightarrow s_1$ if $s_1$ was visited in the current history
and \emph{(ii)} $s_2 \rightarrow s_2$ if $s_1$ was not visited in the current 
histoy. Such strategy is cs-winning but not c-perfect,
due to the c-winning path $s_0 s_2$.}
\label{fig-cwin-not-strongly}
\end{figure}
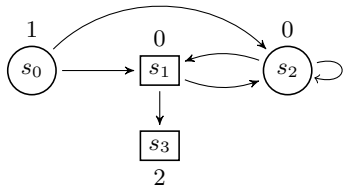
\begin{figure}[h]
\centering
\begin{tikzpicture}[node distance=1.5cm, auto, shorten >=2pt, shorten <=2pt, bend angle=20]
      \node[p1] (s0) [label=above:$0$]              {$s_0$};
      \node[p2] (s1) [label=above:$1$, above right of=s0] {$s_1$};
      \node[p2] (s2) [label=above:$2$, below right of=s0] {$s_2$};
      \node[p1] (s3) [label=above:$0$, below right of=s1] {$s_3$};
      \node[p2] (s4) [label=above:$3$, above right of=s3] {$s_4$};
      \node[p2] (s5) [label=above:$4$, below right of=s3] {$s_5$};

      \path[-stealth'] (s0) edge (s1)
                            edge (s2)
                       (s1) edge (s3)
                       (s2) edge (s3)
                       (s3) edge (s4)
                            edge (s5);
\end{tikzpicture} 
\caption{A game having a positional strongly winning strategy, but no
positional subgame perfect strategy. The goal is 
$\phi = 0 \wedge \bigl( (\Diamond 1 \wedge \Diamond 3) \vee (\Diamond 2 \wedge \Diamond 4) \bigr)$. 
The positional strategy $s_0 \rightarrow s_1$, $s_3 \rightarrow s_4$
is strongly winning.
All subgame perfect strategies need memory in state $s_3$.}
\label{fig-subgame-needs-mem}
\end{figure}
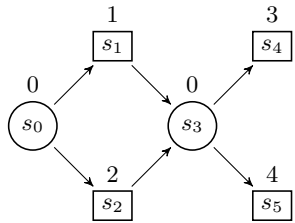


\begin{thebibliography}{dAFMR05}

\bibitem[AHK97]{ATL-FOCS97}
R.~Alur, T.A. Henzinger, and O.~Kupferman.
\newblock Alternating-time temporal logic.
\newblock In {\em Proc. 38th IEEE Symp. Found. of Comp. Sci.}, pages 100--109.
  IEEE Computer Society Press, 1997.

\bibitem[ALW89]{ALW}
M.~Abadi, L.~Lamport, and P.~Wolper.
\newblock Realizable and unrealizable specifications of reactive systems.
\newblock In {\em Proc. 16th Int. Colloq. Aut. Lang. Prog.}, volume 372 of {\em
  Lect. Notes in Comp. Sci.}, pages 1--17. Springer-Verlag, 1989.

\bibitem[BBF07]{trust07}
S.~Baselice, P.A. Bonatti, and M.~Faella.
\newblock On interoperable trust negotiation strategies.
\newblock In {\em POLICY 07: 8th IEEE International Workshop on Policies for
  Distributed Systems and Networks}. IEEE Computer Society, 2007.

\bibitem[BdA95]{bda95}
A.~Bianco and L.~de~Alfaro.
\newblock Model checking of probabilistic and nondeterministic systems.
\newblock In {\em Found. of Software Tech. and Theor. Comp. Sci.}, volume 1026
  of {\em Lect. Notes in Comp. Sci.}, pages 499--513. Springer-Verlag, 1995.

\bibitem[Ber07]{admissibility}
D.~Berwanger.
\newblock Admissibility in infinite games.
\newblock In {\em STACS}, volume 4393 of {\em Lect. Notes in Comp. Sci.}, 2007.

\bibitem[dAFMR05]{emsoft05}
L.~de~Alfaro, M.~Faella, R.~Majumdar, and V.~Raman.
\newblock Code aware resource management.
\newblock In {\em {EMSOFT 05}: 5th Intl.\ ACM Conference on Embedded Software},
  pages 191--202. {ACM} Press, 2005.

\bibitem[dAH01]{EMSOFT01}
L.~de~Alfaro and T.A. Henzinger.
\newblock Interface theories for component-based design.
\newblock In {\em {EMSOFT 01}: 1st Intl.\ Workshop on Embedded Software},
  volume 2211 of {\em Lect. Notes in Comp. Sci.}, pages 148--165.
  Springer-Verlag, 2001.

\bibitem[DJW97]{howmuch}
S.~Dziembowski, M.~Jurdzi{\'n}ski, and I.~Walukiewicz.
\newblock How much memory is needed to win infinite games?
\newblock In {\em Proc. 12th IEEE Symp. Logic in Comp. Sci.} IEEE Computer
  Society, 1997.

\bibitem[EJ91]{EmersonJutla91}
E.A. Emerson and C.S. Jutla.
\newblock Tree automata, mu-calculus and determinacy (extended abstract).
\newblock In {\em Proc. 32nd IEEE Symp. Found. of Comp. Sci.}, pages 368--377.
  IEEE Computer Society Press, 1991.

\bibitem[FV97]{FilarVrieze97}
J.~Filar and K.~Vrieze.
\newblock {\em Competitive {Markov} Decision Processes}.
\newblock Springer-Verlag, 1997.

\bibitem[GZ05]{gimbert05}
H.~Gimbert and W.~Zielonka.
\newblock Games where you can play optimally without any memory.
\newblock In {\em {CONCUR 05}: Concurrency Theory. 16th Int.\ Conf.}, volume
  3653 of {\em Lecture Notes in Computer Science}. Springer, 2005.

\bibitem[KR01]{kremer01}
S.~Kremer and J.-F. Raskin.
\newblock A game-based verification of non-repudiation and fair exchange
  protocols.
\newblock In {\em Proceedings of the 12th International Conference on
  Concurrency Theory}, volume 2154 of {\em Lecture Notes in Computer Science},
  pages 551--565. Springer, 2001.

\bibitem[MPS95]{MalerGames}
O.~Maler, A.~Pnueli, and J.~Sifakis.
\newblock On the synthesis of discrete controllers for timed systems.
\newblock In {\em Proc. of 12th Annual Symp. on Theor. Asp. of Comp. Sci.},
  volume 900 of {\em Lect. Notes in Comp. Sci.} Springer-Verlag, 1995.

\bibitem[OR94]{OsborneRubinstein}
M.J. Osborne and A.~Rubinstein.
\newblock {\em A Course in Game Theory}.
\newblock MIT Press, 1994.

\bibitem[PR89]{PnueliRosner89}
A.~Pnueli and R.~Rosner.
\newblock On the synthesis of a reactive module.
\newblock In {\em Proceedings of the 16th Annual Symposium on Principles of
  Programming Languages}, pages 179--190. ACM Press, 1989.

\bibitem[Tho95]{Thomasgames95}
W.~Thomas.
\newblock On the synthesis of strategies in infinite games.
\newblock In {\em Proc. of 12th Annual Symp. on Theor. Asp. of Comp. Sci.},
  volume 900 of {\em Lect. Notes in Comp. Sci.}, pages 1--13. Springer-Verlag,
  1995.

\bibitem[Zie98]{Zielonka98}
W.~Zielonka.
\newblock Infinite games on finitely coloured graphs with applications to
  automata on infinite trees.
\newblock {\em Theoretical Computer Science}, 200:135--183, June 1998.

\end{thebibliography}
\end{document}